\newif\ifpublic
\publicfalse 

\newif\ifconf
\conffalse 

\newif\ifarxiv
\arxivtrue

\ifconf
\documentclass[a4paper,UKenglish,numberwithinsect]{lipics}
\else
\documentclass[11pt]{article}

\usepackage{ttpalatino}
\usepackage[margin=1in]{geometry}

\usepackage{amsmath}
\usepackage{amssymb}
\usepackage{amsfonts}
\usepackage{graphicx}
\usepackage{fullpage}
\usepackage{pdfsync}
\usepackage{underscore}  
\usepackage{amsthm}
\usepackage{xcolor}

\fi

\ifconf
\else
\usepackage[bookmarks,colorlinks,breaklinks]{hyperref}  

\hypersetup{linkcolor=blue,citecolor=blue,filecolor=blue,urlcolor=blue} 
\fi

\newcommand{\lref}[2][]{\hyperref[#2]{#1~\ref*{#2}}}
\renewcommand{\eqref}[1]{\hyperref[#1]{(\ref*{#1})}}
\numberwithin{equation}{section}

\ifconf
\theoremstyle{plain}
\newtheorem{conjecture}{Conjecture}[section]
\newtheorem{claim}[section]{Claim}
\else
\theoremstyle{plain}
\newtheorem{lem}{Lemma}[section]
\newtheorem{theorem}[lem]{Theorem}
\newtheorem{lemma}[lem]{Lemma}

\newtheorem{corollary}[lem]{Corollary}
\newtheorem{conjecture}[lem]{Conjecture}
\newtheorem{claim}[lem]{Claim}

\newtheorem{definition}[lem]{Definition}

\theoremstyle{definition}
\newtheorem{remark}[lem]{Remark}

\fi

\ifpublic
\usepackage[disable]{todonotes}
\else
\usepackage[colorinlistoftodos]{todonotes}
\fi

\DeclareMathOperator{\poly}{poly}

\DeclareMathOperator{\Inf}{Inf}
\DeclareMathOperator*{\E}{\mathbb{E}}

\newcommand{\etal}{{\em et~al.}}
\newcommand{\F}{\mathbb{F}}
\newcommand{\Ef}{\mathfrak{F}}

\newcommand{\C}{\mathbb{C}}

\newcommand{\R}{\mathbb{R}}
\newcommand{\Pe}{\mathsf{P}}

\renewcommand{\Re}{\mathsf{Re}}

\newcommand{\TwoToOne}{2-to-1}
\renewcommand{\cal}{\mathcal}

\renewcommand{\epsilon}{\varepsilon}
\renewcommand{\phi}{\varphi}

\newcommand{\prob}[2]{\Pr_{#1}\left[#2\right]}
\newcommand{\avg}[2]{\mathop{\mathbb{E}}_{#1}\left[#2\right]}

\newcommand {\email} [1] {Email: \texttt{#1}.}

\ifconf

\title{Derandomized Graph Product Results Using the Low Degree Long Code}
\titlerunning{Derandomized Graph Product Results Using the Low Degree Long Code} 

\author[1]{Irit Dinur\footnote{Irit Dinur's research is supported by ERC-Stg grant number 239985. }}
\author[2]{Prahladh Harsha}
\author[3]{Srikanth Srinivasan}
\author[4]{Girish Varma\footnote{Girish Varma's research is supported by 
    Google Ph.D. Fellowship in Algorithms. Part of
the work was done when the author was visiting the Weizmann Institute of Science, Israel. }}
\affil[1]{Weizmann Institute of Science, Israel. \\   \texttt{irit.dinur@weizmann.ac.il}}
\affil[2]{Tata Institute of Fundamental Research, India.\\ \texttt{prahladh@tifr.res.in}}
\affil[3]{Department of Mathematics, IIT Bombay, India.\\  \texttt{srikanth@math.iitb.ac.in}}
\affil[4]{Tata Institute of Fundamental Research,  India.\\ \texttt{girishrv@tifr.res.in}}

\authorrunning{I. Dinur, P. Harsha, S. Srinivasan, and G. Varma}

\Copyright{Irit Dinur, Prahladh Harsha, Srikanth Srinivasan, and Girish Varma}


\serieslogo{}
\volumeinfo
  {Billy Editor and Bill Editors}
  {2}
  {Conference title on which this volume is based on}
  {1}
  {1}
  {1}
\EventShortName{}
\DOI{10.4230/LIPIcs.xxx.yyy.p}

\begin{document}
\maketitle

\else
\begin{document}
\title{Derandomized Graph Product Results using the Low Degree Long Code}
\author{Irit Dinur\thanks{Weizmann Institute of Science, Israel.
Supported by ERC-Stg grant number 239985. \email{irit.dinur@weizmann.ac.il}.}
\and
Prahladh Harsha\thanks{Tata Institute of Fundamental Research,
  India. \email{prahladh@tifr.res.in}}
\and
Srikanth Srinivasan\thanks{Department of Mathematics, IIT Bombay,
  India. \email{srikanth@math.iitb.ac.in}}
\and Girish Varma\thanks{Tata Institute of Fundamental Research,
  India. Supported by Google India under the Google India PhD
   Fellowship Award. \email{girishrv@tifr.res.in}}
}
\maketitle
\thispagestyle{empty}

\vspace{-0.15in}

\fi

\begin{abstract}
In this paper, we address the question of whether the recent
derandomization results
obtained by the use of the low-degree long code can be extended to
other product settings. We consider two settings: (1) the graph
product results of Alon, Dinur, Friedgut and Sudakov~[{\em GAFA}, 2004]
and (2) the ``majority is stablest'' type of result obtained by Dinur,
Mossel and Regev~[{\em SICOMP}, 2009] and Dinur and
Shinkar~[In {\em Proc. APPROX}, 2010] while studying the hardness of approximate
graph coloring.

In our first result, we show that there exists a considerably smaller
subgraph of $K_3^{\otimes R}$ which exhibits the following property
(shown for $K_3^{\otimes R}$ by Alon~\etal): independent sets close in
size to the maximum independent set are well approximated
by dictators.

The ``majority is stablest'' type of result of Dinur~\etal\ and Dinur
and Shinkar shows that if there exist two sets of vertices $A$ and $B$
in $K_3^{\otimes R}$ with very few edges with one endpoint in $A$ and
another in $B$, then it must be the case that the two sets $A$ and $B$
share a single influential coordinate. In our second result, we show
that a similar ``majority is stablest'' statement holds good for a
considerably smaller subgraph of $K_3^{\otimes R}$. Furthermore using
this result, we give a more efficient reduction from Unique Games
to the graph coloring problem, leading to improved hardness  of
approximation results for coloring.

\end{abstract}

\section{Introduction}

The discovery of the low-degree long code (aka short code) by
Barak~\etal~\cite{BarakGHMRS2012} has over the last year led to
several more efficient inapproximability
reductions~\cite{BarakGHMRS2012,DinurG2013,GuruswamiHHSV2014,KhotS2014b,Varma2014}. The
low-degree long code is a derandomization of the long code in the
following sense. Given a finite field $\F$, the long code of a string
$x \in \F^n$ is the evaluation of every $\F$-valued function on $\F^n$
at the point $x$ while the degree $d$ long code of of $x$ is the
evaluation of every $n$-variate polynomial of total degree at most $d$
at the point $x$. The crucial observation of
Barak~\etal~\cite{BarakGHMRS2012} was that the optimal testing results
for Reed-Muller codes~\cite{BhattacharyyaKSSZ2010,HaramatySS2013}
proved that the low-degree long code could be used as a surrogate for
the long code in several inapproximability results. In this paper, we
ask if we can extend this application of low-degree long code to other
product settings. In particular, we prove the following two
results. (1) We show that result due to Alon~\etal~\cite{AlonDFS2004}
on the size of maximum independent sets in product graphs can be
derandomized (\lref[Theorem]{thm:derand-indep}). (2) We show that the
``majority is stablest'' type of result obtained by
Dinur~\etal~\cite{DinurMR2009} and Dinur and Shinkar~\cite{DinurS2010}
can be derandomized (\lref[Theorem]{thm:derand-ds}).

\subsection{Derandomized graph products}

As a first application, we consider the following graph product result
due to Alon~\etal~\cite{AlonDFS2004}.  Consider the undirected
weighted graph $K_3$ on the three vertices $V=\{0,1,2\}$ and edges
weighted as follows: $W(f,f') = 1/2$ iff $f'\neq f \in \{0,1,2\}$. Let
$K_3^{\otimes R}$ be the graph with vertex set $V^{\otimes R}$ and
weights-matrix the $R$-wise tensor of the matrix $W$. Clearly, for any
$i \in [R]$ and $a \in \{0,1,2\}$, the set $V_{i,a} := \{ v \in
V^{\otimes R} : v_i = a \}$ is an independent set in $K_3^{\otimes R}$
of fractional size $1/3$ since $K_3$ does not have any self loops. We
call such an independent set a \emph{dictator} for obvious
reasons. Alon~\etal~\cite{AlonDFS2004} showed that these are the
maximal independent sets in $K_3^{\otimes R}$ and in fact any
independent set of size close to the maximum is close to a dictator.
\begin{theorem}[\cite{AlonDFS2004}]\label{thm:adfs}
  Let $A$ be an independent set in $K_3^{\otimes R}$ of size $\delta
  3^R$. Then,
\begin{enumerate}
\item $\delta \leq 1/3$.
\item $\delta = 1/3$ iff $A$ is a dictator.
\item If $\delta \geq 1/3 -\epsilon$, then $A$ is $O(\epsilon)$-close
  to a dictator. That is, there is a dictator $A'$ such that $|A\Delta
  A'| = O(\epsilon3^R)$.
\end{enumerate}
\end{theorem}

Note that the above graph has $3^R$ vertices. Our first result
(\lref[Theorem]{thm:derand-indep}) shows that there exists a
considerably smaller subgraph $\cal G = (\cal V, \cal E)$ of
$K^{\otimes R}$ with only $3^{\poly(\log R)}$ vertices that has the
same properties. In order to describe the subgraph, it will be
convenient to think of $K_3$ as having vertex set $\F_3$ and $$W(f,f')
= \Pr_{p \in \F_3,a \in \{1,2\}} [f' = f +a(p^2+1)].$$ 
Let $\Pe_{r,d}$
be the set of polynomials on $r$ variables over $\F_3$ of total degree
at most $d$ and individual degrees of the variables at most $2$. Let
$r$ and $d$ be two parameters and let $R = 3^r$. Note that $V^{\otimes
  R}$ can be identified with $\Pe_{r,2r}$, since $\Pe_{r,2r}$ is the
set of all functions from $\F^r_3$ to $\F_3$. The subgraph $\cal G =
(\cal V, \cal E)$ is as follows : $\cal V := \Pe_{r, 2d}$ and the
edges are given by the weights-matrix defined below
$$\cal W(f,f') = \Pr_{p \in \Pe_{r,d},a \in \{1,2\}} [f' = f +a(p^2+1)].$$
Note that since $\Pe_{r,2d}$ is a subspace of dimension $r^{O(d)}$,
the size of the vertex set is $3^{r^{O(d)}}$, which is considerably
smaller than $3^R$ for constant $d$.
\begin{theorem}
\label{thm:derand-indep}
There is a constant $d$ for which the following holds. If $A$ is an independent set of size $\delta |\cal V|$ in $\cal G$ then
\begin{enumerate}
\item $\delta \leq 1/3$.\label{itm:comp-1}
\item $\delta =1/3$ iff $A$ is a dictator. \label{itm:comp-2}
\item If $\delta \geq 1/3 -\epsilon$ then $A$ is $O(\epsilon)$-close to a dictator.\label{itm:sound}
\end{enumerate}
\end{theorem}

A crucial element in the proof of \lref[Theorem]{thm:adfs} is a
hypercontractivity theorem for functions which do not have any heavy
Fourier coefficients. \lref[Theorem]{thm:derand-indep} is proved by
observing that a similar hypercontractivity theorem also holds good in
the low-degree long code setting (see \lref[Lemma]{thm:hyp-for-poly}).
\subsection{Derandomized ``majority is stablest'' result}

While studying the hardness of approximate graph coloring, Dinur,
Mossel and Regev~\cite{DinurMR2009} proved the following ``majority is
stablest'' type of result: if there is a pair of subsets of vertices
in $K_3^{\otimes R}$ of sufficiently large size such that the  average weight
of edges between them is small, then their indicator functions must
have a common influential coordinate. 
Subsequently, Dinur and Shinkar~\cite{DinurS2010} obtained the
following quantitative improvement to the above theorem.

\begin{theorem}[{\cite[Theorem 1.3]{DinurS2010}}]
\label{thm:ds} 
For all $\mu >0$ there exists $\delta= \mu^{O(1)}$ and $k = O(\log 1/\mu)$ such that the following holds: For any two functions $A,B:\{0,1,2\}^R \rightarrow[0,1]$ if  
$$ \E A> \mu,~  \E B > \mu,~ \text{ and }~\E_{f,f'} A(f) B(f')  \leq
\delta\footnote{The hypothesis in the theorem statement of
  Dinur-Shinkar~\cite{DinurS2010} requires $\E_{f,f'}A(f)B(f') =0$,
  however it is easy to check that their theorem also holds good under
  the weaker hypothesis $\E_{f,f'} A(f) B(f')  \leq
\delta$.}$$
where $f$ is chosen randomly from $V^{\otimes R}$ and $f'$ is chosen with probability $W^{\otimes R}(f,f')$ then 
$$ \exists x \in [R] \text{ such that } \Inf^{\leq k}_x(A) \geq \delta \text{  and } \Inf_x^{\leq k} (B) \geq \delta.$$ 
\end{theorem}

Our second result (\lref[Theorem]{thm:derand-ds}) shows that the above
theorem can be derandomized to obtain a similar result for the
subgraph $\cal G$.  For defining influence for real valued functions
on $\Pe_{r,2d}$, we note that the characters of $\Pe_{r,2d}$ are
restrictions of characters of $\F_3^{ R}\equiv \Pe_{r,2r}$. So the
definition of influence for functions on $\F_3^{ R}$ also extends
naturally to functions on $\Pe_{r,2d}$.
\begin{theorem}\label{thm:derand-ds}
For all $\mu >0$ there exists $\delta= \mu^{O(1)}, k = O(\log 1/\mu), d = O(\log 1/\mu)$ such that the following holds: For any two functions $A,B:\Pe_{r,2d} \rightarrow[0,1]$ if  
$$ \E A> \mu,~  \E B > \mu,~ \text{ and }~\E_{f,f'} A(f) B(f')  \leq \delta$$
where $f$ is chosen randomly from $\Pe_{r,2d}$, $f' =f + a(p^2 + 1)$, $p$ are chosen randomly from $\Pe_{r,d}$ and $a \in_R \{1,2\}$ then 
$$ \exists x \in \F_3^r \text{ such that } \Inf^{\leq k}_x(A) \geq \delta \text{  and } \Inf_x^{\leq k} (B) \geq \delta.$$ 
\end{theorem}

A similar derandomized ``majority is stablest'' result in the case of
the noisy hypercube was proved by
Barak~\etal~\cite[Theorem~5.6]{BarakGHMRS2012} and they used the
Meka-Zuckerman pseudorandom generators (PRGs) for polynomial threshold
functions~\cite{MekaZ2013}. Kane and Meka~\cite{KaneM2013} obtained a
quantitative improvement over this derandomization by constructing an
improved PRG for Lipschitz functions. Our setting is slightly more
involved, (1) we have a two function version (ie., $A$ and $B$) and
(2) the underlying graph in $K_3$ and the corresponding noise operator 
in the derandomized setting has not necessarily positive eigenvalues. Yet, we manage to
show that a derandomization still holds in this case too (using the
Kane-Meka PRG). \ifarxiv \else We conjecture that our derandomization can be further
improved to obtain $d = O(\log \log 1/\mu)$. 
\fi

\subsubsection{Application to graph coloring}

Using a version of \lref[Theorem]{thm:ds} for another base graph on $4$ vertices, Dinur and Shinkar proved a hardness result for graph coloring. 
\begin{definition}[Label Cover]
\label{def:label-cover}
An instance $G=(U,V,E,L, R,\{\pi_e\}_{e\in E})$ of a \emph{Label Cover} consists of a bipartite graph $(U,V,E)$  that is right regular along with a projection map $\pi_e : R \rightarrow L$ for every edge $e\in E$. Label Cover is a constraint satisfaction problem where the vertices in $U$ are the variables taking values in $L$ and vertices in $V$ taking values in $R$. The instance is a \emph{Unique Games} instance if $R=L$ and $\pi_e$ is a permutation for all $e \in E$. Given a labeling $\ell : U \cup V \rightarrow L\cup  R$, an edge $e = (u,v)$ is said to be satisfied if $\pi_e(\ell(v)) = \ell(u)$. 
\end{definition}

Dinur and Shinkar gave a reduction from an instance of Label Cover  with $n$ vertices, \TwoToOne\ constraints and label set of size $R$ to a graph of size $n 4^R$. Perfectly satisfiable instances were mapped to $4$-colorable graphs. Instances  for which any labeling can satisfy only an $s(n)$ fraction of edges where mapped to graphs which did not have any independent sets of size $\poly(s(n))$. Since the size of the graph produced by the reduction is exponential in $R$, they needed to assume that $R= O(\log n)$, to get hardness results. We give a more efficient reduction using \lref[Theorem]{thm:derand-ds} from Label Cover instances for which the projection constraints have special form. Our reduction is simpler to describe for the case $3$-colorable graphs and starts with Unique Games instances.
Hence for getting hardness result, we need to assume a conjecture similar to the Unique Games Conjecture with specific parameters.

\begin{conjecture}[$(c(n),s(n),r(n))$-UG Conjecture]
It is NP-Hard to distinguish between unique games instances $(U,V,E,R,\Pi)$ on $n$ vertices and $R=\F_3^{r(n)}$ from the following cases:
\begin{itemize}
\item YES Case : There is a labeling and a set $S\subseteq V$ of size $(1-c(n))|V|$ such that all edges between vertices in $S$ are satisfied.
\item NO Case : For any labeling, at most $s(n)$ fraction of edges are satisfied.
\end{itemize}
\end{conjecture}

Khot and Regev~\cite{KhotR2008} proved that the Unique Games Conjecture implies that for any constants $c, s \in (0,1/2)$ there is  a constant $r$ such that $(c,s,r)$-UG Conjecture is true. We also require that the constraints of the Unique Games instance are full rank linear maps.

\begin{definition}[Linear constraint]
A constraint $\pi:R \rightarrow L$ is a \emph{linear constraint} of  iff $R = L=\F_3^r$, and $\pi$ is a linear map of rank $r$.
\end{definition}
The theorem below is obtained by replacing the long code by the low
degree long code of degree $d=O(\log 1/\mu)$ in the reduction of Dinur
and Shinkar. 
\ifarxiv
\else
For want of space, the description of this reduction is
deferred to the full version~\cite[Appendix A]{DinurHSV2014}.
\fi
\begin{theorem}\label{thm:col-hard}
 There is a reduction from $(c,s,r)$-Unique Games instances $G$  with $n$ vertices, label set $\F_3^r$ and linear constraints to graphs $\cal G$ of size $n3^{r^{O(\log 1/\mu)}}$ where $\mu = \poly(s)$ such that
\begin{itemize}
\item If $G$ belongs to the YES case of $(c,s,r)$-UG Conjecture then there is a subgraph of $\cal G$ with fractional size $1-c$ that is $3$-colorable.
\item If $G$ belongs to the NO case of $(c,s,r)$-UG Conjecture then $\cal G$ does not have any independent sets of fractional size $\mu$.
\end{itemize}
\end{theorem}

Due to the improved efficiency of the reduction, we are able to get hardness
results even if the label cover instances have super-polylogarithmic
sized label sets of size at most $2^{2^{O\left(\sqrt{\log \log n}\right)}}$, while
the reduction due to Dinur and Shinkar only works if the label set is of
size at most $O(\log^c n)$ for some constant $c$. More precisely, suppose the UG
conjecture were true for soundness $s(n)$
and alphabet size $R=3^r$ that satisfy $\log_3 R= r= s(n)^{O(1)}$.
Then, the result of Dinur and Shinkar rules out polynomial time algorithms
that find an independent set of relative size $1/\poly(\log \log
N)$. On the other hand, under the same assumption, our reduction rules out polynomial time algorithms
that find an independent set of relative size $1/2^{\poly(\log \log N)}$.


\begin{corollary}\label{cor:main}
Let $c,s,r$ be functions such that $r(n) = \poly(1/s(n))$.
Assuming $(c,s,r)$-UG Conjecture on instances with linear constraints,  given a graph on $N$ vertices which has an induced subgraph of relative size $1-c$ that is $3$-colorable, no polynomial time algorithm can find an independent set of fractional size $2^{-\poly(\log \log N)}$.
\end{corollary}

We remark that we can improve the
conclusion if \lref[Theorem]{thm:derand-ds} can be proved even when $d
= O(\log \log 1/\mu)$.

\section{Preliminaries}
\subsection{Low degree polynomials}
We will be working over the field $\F_3$. Let $\Pe_{r,d}$ be the set of degree
$d$ polynomials on $r$ variables over $\F_3$, with individual variable degrees at most $2$. Let $\mathfrak F_r :=
\Pe_{r,2r}$. Note that $\mathfrak F_r$ is the set of all
functions from $\F_3^r$ to $\F_3$. $\mathfrak F_r$ is a $\F_3$-vector
space of dimension $3^r$ and $\Pe_{r,d}$ is a subspace of dimension
$r^{O(d)}$. The Hamming distance between $f$ and $g \in \mathfrak F_r$,
denoted by $\Delta(f,g)$, is the number of inputs on which $f$ and $g$
differ. For $S \subseteq \mathfrak F_r$, define $\Delta(f, S) := \min_{g\in
  S} \Delta(f,g)$. 
We say that $f$ is $\delta$-far from $S$ if 
$\Delta(f,S) \geq \delta$ and $f$ is $\delta$-close to $S$ otherwise. Given $f,g, \in \mathfrak F_r$, the dot product
between them is defined as $\langle f,g \rangle := \sum_{x \in \F^r}
f(x)g(x)$. For a subspace $S \subseteq \mathfrak F_r$, the dual
subspace is defined as $S^{\perp} := \{ g \in \mathfrak F_r : \forall f
\in S, \langle g,f \rangle = 0 \}$. 
The following theorem relating
dual spaces is well known.
\begin{lemma} 
\label{lem:dual}$\Pe_{r,d}^\perp = \Pe_{r,2r-d-1}$.
\end{lemma}
\noindent We need the following Schwartz-Zippel-like Lemma for degree $d$ polynomials over $\F_3$.
\begin{lemma}[Schwartz-Zippel lemma~{\cite[Lemma~3.2]{HaramatySS2013}}]
\label{lem:SZ}
Let $f\in \F_3[x_1,\cdots,x_r]$ be a non-zero polynomial of degree at most $d$ with individual degrees at most $2$. Then $\prob{a\in \F_3^r}{f(a)\neq 0} \geq 3^{-d/2}$.
\end{lemma}
The following lemma is an easy consequence of \lref[Lemma]{lem:SZ}.
\begin{lemma}
\label{lem:low-deg-local-ind}
If $p$ is a uniformly random polynomial from $\Pe_{r,d}$ then as a
string of length $3^r$ over the alphabet $\F_3$, $p$ is $3^{\ lfloor (d+1)/2 \rfloor}$-wise independent. 
\end{lemma}

\subsection{Fourier analysis of functions on subspace of low degree polynomials}

\begin{definition}[Characters]
A character of $\Pe_{r,d}$ is a function $\chi:\Pe_{r,d} \rightarrow \C$ such that
$$\chi(0) = 1 \text{ and } \forall f,g \in \Pe_{r,d},~ \chi(f+g) = \chi(f)\chi(g).$$
\end{definition}
\noindent The following lemma lists the basic properties of characters.
\begin{lemma}
\label{lem:fourier}
Let $\{1,\omega,\omega^2\}$ be the cube roots of unity and
for $\beta \in \mathfrak F_r, f \in \Pe_{r,d}$, $\chi_\beta(f) :=
\omega^{\langle \beta, f \rangle}$, where $\langle\beta, f\rangle := \sum_{x\in \F_3^r}\beta(x)f(x)$.
\begin{itemize}
\item The characters of $\Pe_{r,d}$ are $\{ \chi_\beta : \beta \in  \mathfrak F_r\}$. 
\item For $\beta \in \Pe_{r,d}^\perp$, $\chi_\beta$ is the constant $1$ function.
\item For any $\beta,\beta' \in \mathfrak F_r$, $\chi_\beta = \chi_\beta'$ if and only if  $\beta-\beta' \in \Pe_{r,d}^\perp$.
\item For any $\beta$, let $|\beta|$ be the size of the set of inputs on which $\beta$ is non-zero. For any distinct $\beta,\beta' \in \mathfrak F_r$ with $|\beta|,|\beta'| < 3^{\lfloor (d+1)/2 \rfloor}/2$,  $\chi_\beta \neq \chi_\beta'$  since $\beta +\beta' \notin \Pe_{r,d}^\perp$.

\item\label{item:minsup} $\forall \beta, \exists \beta'$ such that $\beta-\beta' \in
  \Pe_{r,d}^\perp$ and $|\beta' | = \Delta(\beta,
  \Pe_{r,d}^\perp)$ (i.e., the constant $0$ function is (one of) the
  closest function to $\beta'$ in $\Pe_{r,d}^\perp$). We call such a
  $\beta'$ a minimum support function for the coset $\beta + \Pe_{r,d}^\perp$.

\item Characters forms an orthonormal basis for the vector space of functions from $\Pe_{r,d}$ to $\C$, under the inner product  $\langle A, B\rangle := \E_{f \in \Pe{r,d}} \left[A(f)\overline{B(f)}\right]$
\item Any function $A:\Pe_{r,d} \rightarrow \C$ can be uniquely decomposed as
\begin{equation}
\label{eqn:fourier-decomposition}
A(f) = \sum_{\beta \in \Lambda_{r,d}} \widehat{A}(\beta) \chi_\beta(f)
\text{ where } \widehat{A}(\beta) := \E_{g \in \Pe_{r,d}} \left[A(g)
  \overline{\chi_\beta(g)}\right],
\end{equation} 
and $\Lambda_{r,d}$ is the set of minimum support functions, one for
each of the cosets in $\mathfrak F_r/\Pe_{r,d}^\perp$, with ties broken arbitrarily.
\item Parseval's identity: For any function $A:\Pe_{r,d} \rightarrow\C$,
\begin{equation}
\label{eqn:parseval-general}
\sum_{\beta \in \Lambda_{r,d}} |\widehat A(\beta)|^2 = \E_{f\in \Pe_{r,d}} [|A(f)|^2].
\end{equation}
 In particular, if $A:\Pe_{r,d} \rightarrow \{1,\omega,\omega^2\}$, 
\begin{equation}
\label{eqn:parseval}
\sum_{\beta\in \Lambda_{r,d}}|\widehat A(\beta)|^2 =1.
\end{equation}
\end{itemize}
\end{lemma}

\begin{definition}[Influence]
For a function $A:\Pe_{r,d} \rightarrow \C$ and a number $k < 3^{\lfloor (d+1)/2 \rfloor}/2$, the degree $k$ influence of $a\in \F_3^r$ is defined as
$$\Inf^{\leq k}_a(A) = \sum_{\beta \in \Lambda_{r,d}:\beta(a) \neq 0 \text{ and } |\beta| \leq k} |\widehat A(\beta)|^2.$$
\end{definition}

\begin{definition}[Dictator]
A function $A:\Pe_{r,d} \rightarrow \C$ is a dictator if there exists $x\in \F_3^r$ and $\widehat A_0, \widehat A_{1}, \widehat A_{2} \in \C$ such that $A$ can be written as $A = \widehat A_0 + \widehat A_{1}\chi_{e_x} +\widehat A_{2}\chi_{2e_x}$ where $e_x:\F_3^r \rightarrow \F_3$ the indicator function for $x$.
\end{definition}

The following lemma which follows from the results of Guruswami \etal\ \cite{GuruswamiHHSV2014}, will be crucial for our proofs.
\begin{lemma}
\label{lem:cor-rand-square}
If $\alpha:\F_3^r \rightarrow \F_3$ such that $\Delta(\alpha,\Pe^{\perp}_{r,2d}) > 3^{d/2}$ then
$$ \left|\E_{p \in \Pe_{r,d}} \chi_\alpha(p^2) \right| \leq 3^{-\Omega(3^{d/9})}.$$
\end{lemma}
\begin{proof}
By definition, $\left|\E_{p \in \Pe_{r,d}} \chi_\alpha(p^2) \right| = \left|\E_{p \in \Pe_{r,d}} \omega^{\langle\alpha, p^2\rangle} \right|$. If $\alpha:\F_3^r \rightarrow \F_3$ is such that $\Delta(\alpha,\Pe^\perp_{r,2d}) > 3^{d/2}$ then for a random $ p \in \Pe_{r,d}$,
$\langle \alpha , p^2 \rangle$ is $3^{-\Omega(3^{d/9})}$-close to the uniform distribution on $\F_3$ according to {\cite[Lemma 3.1 and 3.4]{GuruswamiHHSV2014}}. 
\end{proof}
\ifarxiv
\else
\pagebreak
\fi
\section{\texorpdfstring{Derandomized {$K_3^{\otimes R}$}}%
                   {Derandomized Product K3}}

Alon \etal\ \cite{AlonDFS2004} proved \lref[Theorem]{thm:adfs} by using the following lemma.
\begin{lemma}
\label{thm:dict}
There is constant $K$ such that the following holds: If $A:\F_3^R \rightarrow \{0,1\}$ satisfies
 $$\sum_{|\alpha| > 1} |\widehat A_\alpha|^2 \leq \epsilon \text{ and } \widehat A_0 =\delta$$ 
 then there exists a dictator $B:\F_3^R \rightarrow \{0,1\}$ such that 
$$\|A-B\|_2 \leq \frac{K\epsilon}{\delta -\delta^2 -\epsilon}.$$
\end{lemma}
The above lemma was proved using  the following hypercontractive inequality.
\begin{lemma}
\label{thm:hyp}
There is a constant $C$ such that for any function $A:\F_3^R \rightarrow \C$ with $\widehat A_\alpha = 0$ when $|\alpha| > t$,
$$ \|A\|_4 \leq C^t \|A\|_2.$$
\end{lemma}
Our proof of \lref[Theorem]{thm:derand-indep} will use a similar lemma for functions on the subspace $\Pe_{r,2d}$. 
\begin{lemma}
\label{thm:derand-dict}
There is a constant $K$ such that the following holds: If $A:\Pe_{r,2d} \rightarrow \{0,1\}$ satisfies
 $$\sum_{|\alpha| > 1} |\widehat A_\alpha|^2 \leq \epsilon \text{ and } \widehat A_0 =\delta$$ 
 then there exists a dictator $B:\Pe_{r,2d} \rightarrow \{0,1\}$ such that 
$$\|A-B\|_2 \leq \frac{K\epsilon}{\delta -\delta^2 -\epsilon}.$$
\end{lemma}
The above lemma follows from hypercontractive inequalities over $\Pe_{r,2d}$ stated below, in exactly the same way as Alon \etal\ proves \lref[Lemma]{thm:dict} from \lref[Lemma]{thm:hyp}.
\begin{lemma}
\label{thm:hyp-for-poly}
There is a constant $C$ such that for $4t\leq 3^{d-1}$ and any function $A:\Pe_{r,2d} \rightarrow \C$ with $\widehat A_\alpha = 0$ when $|\alpha| > t$,
$$\|A\|_4 \leq C^t \|A\|_2.$$
\end{lemma}
\begin{proof}
Follows from \lref[Lemma]{thm:mov-to-poly} and \lref[Lemma]{thm:hyp}.
\end{proof}
\begin{definition}[Lift]\label{def:lift}
For a function $B:\Pe_{r,2d} \rightarrow \C$ with the Fourier decomposition $B= \sum_{\alpha \in \Lambda_{r,d}} \widehat{B}_\alpha \chi_\alpha$, the lift of $B$ denoted by $B'$ is a function $B':\Ef_r \rightarrow \C$ with the Fourier decomposition $B'= \sum_{\alpha \in \Lambda_{r,d}} \widehat{B}_\alpha \chi_\alpha$. In the decomposition of $B'$, $\chi_\alpha$'s are functions with domain $\Ef_r$.
\end{definition}
\begin{lemma}
\label{thm:mov-to-poly}
If $2kt \leq 3^{d-1}$ and  $B:\Pe_{r,2d} \rightarrow \C$ be a function such that $\widehat B_\alpha = 0$ when $|\alpha| > t$ then
$$\|B\|_{2k} = \|B'\|_{2k}.$$
\end{lemma}
\begin{proof}
From the \lref[Lemma]{lem:SZ} and \lref[Lemma]{lem:dual}, we have that $\forall \alpha \in \Pe^{\perp}_{r,2d}\setminus \{0\},~|\alpha| > 3^{d-1}$. So if  $\exists \{\alpha_i,\beta_i \}_{i \in [k]}$ with $|\alpha_i|,|\beta_i| \leq t$, then
\begin{equation}
\label{eqn:small-sup-is-zero}
\sum_{i \in [k]} \alpha_i -\beta_i \in \Pe^{\perp}_{r,2d} \Rightarrow \sum_{i \in [k]} \alpha_i -\beta_i = 0.
\end{equation}
This is because $\sum_{i \in [t]} \alpha_i -\beta_i$ has support size at most $2kt < 3^{d-1}$. We use this fact to prove the theorem as follows:
\begin{align*}
\|B\|_{2k}^{2k} & = \E_{f \in \Pe_{r,2d}} |B(f)|^{2k} = \E_{f \in \Pe_{r,2d}} \prod_{i \in [k]}B(f)\overline {B(f)}\\
&= \sum_{\alpha_1,\beta_1,\cdots,\alpha_k,\beta_k \in \Lambda_{n,2d}} \left(\prod_{i \in [k]}\widehat B_{\alpha_i} \overline{\widehat B_{\beta_i}}\right) \E_{f \in \Pe_{r,2d}} \prod_{i \in [k]} \chi_{\alpha_i}(f) \overline{\chi_{\beta_i}(f)}~~(~\text{from } \eqref{eqn:fourier-decomposition}~)\\
&= \sum_{\substack{\alpha_1,\beta_1,\cdots,\alpha_k,\beta_k \in \Lambda_{r,2d}\\ \sum_i \alpha_i -\beta_i \in \Pe^{\perp}_{r,2d}}}~ \prod_{i \in [k]}\widehat B_{\alpha_i} \overline{\widehat B_{\beta_i}} \\
&= \sum_{\substack{\alpha_1,\beta_1,\cdots,\alpha_k,\beta_k \in \Lambda_{r,2d} \\ \sum_i \alpha_i -\beta_i = 0}}~ \prod_{i \in [k]}\widehat B_{\alpha_i} \overline{\widehat B_{\beta_i}}~~(\text{ from } \eqref{eqn:small-sup-is-zero}~) \\
&= \sum_{\alpha_1,\beta_1,\cdots,\alpha_k,\beta_k\in \Lambda_{r,2d}} \left(\prod_{i \in [k]}\widehat B_{\alpha_i} \overline{\widehat B_{\beta_i}}\right) \E_{f \in \Ef_r} \prod_{i \in [k]} \chi_{\alpha_i}(f) \overline{\chi_{\beta_i}(f)}\\
&= \E_{f \in \Ef_r} \prod_{i \in [k]}B'(f)\overline {B'(f)}= \E_{f \in \Ef_r} |B'(f)|^{2k}= \|B'\|_{2k}^{2k}
\end{align*}
\end{proof}

\subsection{Proof of \texorpdfstring{\lref[Theorem]{thm:derand-indep}}{Theorem 1.2}}
\begin{proof}[Proof of \ref{itm:comp-1}]
For $f\in V$, consider the set $\{f,f+1,f+2 \} \subseteq V $. These sets form a partition of $V$ and are triangles in the graph. Hence $\delta \leq 1/3$.
\end{proof}
\begin{proof}[Proof of \ref{itm:comp-2}]
Let $A:\Pe_{r,2d} \rightarrow \{0,1\}$ be the indicator set of the independent set of size $\delta |V|$. By Parseval's equation~\eqref{eqn:parseval-general} and the fact that $\widehat A_0 = \delta$, we have that
\begin{equation}
\label{eqn:parseval-1}
\sum_{\alpha \in \Lambda_{r,2d} \setminus \{ 0\}} |\widehat A_\alpha|^2 = \delta -\delta^2.
\end{equation}
Since $A$ is an independent set, 
$$\E_{p \in \Pe_{r,d},a \in \F_3, f\in \Pe_{r,2d}} A(f)A(f+a(p^2+1)) = \sum_{\alpha\in \Lambda_{r,2d}} |\widehat A_\alpha|^2 \E_{p\in \Pe_{r,d}, a \in \F_3}\chi_\alpha(a(p^2+1)) = 0.$$
Taking the real parts of the equation on both sides and rearranging, we get
\begin{equation}
\label{eqn:four-indep-1}
 \sum_{\alpha \in \Lambda_{r,2d} \setminus \{ 0\}} |\widehat A_\alpha|^2 \Re\left(\E_{p\in \Pe_{r,d}}\chi_\alpha(p^2+1)\right) = -\delta^2.
\end{equation}
Let $T$ be a random variable such that $\Pr[T=\alpha] = |\widehat A_\alpha |^2/(\delta-\delta^2)$ and $X$ be the random variable $X(T) = \Re\left(\E_{p\in \Pe_{r,d}, a \in \F_3}\chi_\alpha(a(p^2+1))\right)$. From \eqref{eqn:parseval-1} and \eqref{eqn:four-indep-1}, we have that
$$ \E X = \frac{-\delta}{1-\delta}.$$
Since $p$ is a random degree $d$ polynomial, it is $3^{d/2}$-wise independent from \lref[Lemma]{lem:low-deg-local-ind}.  So if $|T| \leq 3^{d/2}$ then
\begin{align*}
 &\left| \Re\left(\E_{p\in \Pe_{r,d}, a \in \F_3}\chi_\alpha(a(p^2+1))\right) \right|\\
 &= \left| \frac{1}{2}\Re\left(\left(\frac{\omega^2-1}{3}\right)^{|\alpha|_1}\left(\frac{\omega-1}{3}\right)^{|\alpha|_2} + \left(\frac{\omega-1}{3}\right)^{|\alpha|_1}\left(\frac{\omega^2-1}{3}\right)^{|\alpha|_2} \right) \right| \leq \left(\frac{1}{\sqrt{3}}\right)^{|\alpha|}
\end{align*}
where $|\alpha|_a = \{ x \in \F_3^r: \alpha(x) = a \}$.

If $|T| > 3^{d/2}$, we know from \lref[Lemma]{lem:cor-rand-square} that $|X(T)| \leq 3^{-\Omega(3^{d/9})}$.

Note that for $T$ with $|T|=1$, $X(T) = -1/2$. For $T$ with $|T|=2$, $ X(T) \geq 0$. For $T$ with $|T| \geq 3, X(T) \geq \frac{-1}{3\sqrt{3}}$. So if $\E X = -1/2$ then $\Pr[|T|=1] = 1$. So $A$ is a Boolean valued function with non zero Fourier coefficients of supports only $0$ and $1$. Using arguments similar to Proof of  \cite[Lemma 2.3]{AlonDFS2004}, it can be shown that there is an $x\in \F_3^r$ such that $A(f)$ only depends on $f(x)$ .

\end{proof}
\begin{proof}[Proof of \ref{itm:sound}]
Suppose $\delta= 1/3 -\epsilon$. First we show that most of Fourier weights are concentrated in the first two levels

\begin{lemma}
\label{lem:low-fourier-conc}
$$\sum_{\alpha \in \Lambda_{r,2d}: |\alpha| > 1} |\widehat A_\alpha|^2 \leq 2\epsilon$$
\end{lemma} 
\begin{proof}
Consider the random variables $X$ and $T$ defined in the Proof of \ref{itm:comp-2}. Since $\delta = 1/3 -\epsilon$ and since $\epsilon < 1/3$, $\E X = -1/2 + \epsilon$. Let $Y$ be the random variable $X+ 1/2$. Note that $Y\geq 0$ and when $Y >0$, $Y \geq 1/6$. Therefore by Markov, $\prob{}{ Y > 0} \leq 6\epsilon$ and 
 
$$ \sum_{\alpha \in \Lambda_{r,2d}: |\alpha| > 1} |\widehat A_\alpha|^2 \leq (\delta -\delta^2) \prob{}{Y>0} \leq 2\epsilon.$$
\end{proof}

Then we use \lref[Lemma]{thm:derand-dict} to obtain the result.

\end{proof}

\section{Derandomized Majority is Stablest}

In this section, we prove \lref[Theorem]{thm:derand-ds}. The graphs described in \lref[Theorem]{thm:derand-ds} and \lref[Theorem]{thm:ds} can be viewed as  Cayley graphs on a suitable group. For the proof, we will need bounds on the eigenvalues of these Cayley graphs.  For a group $G$, $\R^G$ denotes the vector space of real valued functions on $G$.

\begin{definition}[Cayley Operator]
For a group $G$ with operation $+$, an operator $M:\R^G \rightarrow \R^G$ is a \emph{Cayley operator} if there is a distribution $\mu$ on $G$ such that for any function $A:G\rightarrow \R$,
$$ (MA)(f) = \E_{\eta \in \mu} A(f+\eta).$$
It is easy to see that a character $\chi:G \rightarrow \C$ is an eigenvector of $M$ with eigenvalue $\E_{\eta \in \mu} \chi(\eta)$.
\end{definition}

\begin{definition}
\label{def:ops}
We define the following Cayley operators:
\begin{enumerate}
\item For the group $\F_3$, let $T:\R^{\F_3} \rightarrow \R^{\F_3}$ be the Cayley operator corresponding to the distribution $\mu$ that is uniform on $\F_3\setminus \{0\}$. Let $\lambda$ be the second largest eigenvalue in absolute value of $T$.
\item For the group $\Ef_r$, let $T_r: \R^{\Ef_r} \rightarrow \R^{\Ef_r}$ be the Cayley operator corresponding to the distribution $\mu_r$ that is uniform on $\{ f \in \Ef_r : f^{-1}(0) = \emptyset\}$. Let $\lambda_r(\alpha)$ be the eigenvalue of $T_r$ corresponding to the eigenvector $\chi_\alpha$, for $\alpha \in \Ef_r$.
\item For the group $\Pe_{r,2d}$, let $T_{r,d}: \R^{\Pe_{r,2d}} \rightarrow \R^{\Pe_{r,2d}}$ be the Cayley operator corresponding to the distribution $\mu_{r,2d}$ of choosing a uniformly random element in $\{ p^2+1, -p^2 - 1\}$ where $p \in  \Pe_{r,2d}$ is chosen uniformly at random. Let $\lambda_{r,d}(\alpha)$ be the eigenvalue of $T_{r,d}$ corresponding to $\chi_\alpha$, for $\alpha \in \Ef_r$.
\item For the group $\Pe_{r,2d}$, let $S_{r,d}: \R^{\Pe_{r,2d}} \rightarrow \R^{\Pe_{r,2d}}$ be the Cayley operator corresponding to the distribution of $a \cdot \prod_{i=1}^d (\ell_i-1)(\ell_i-2)$ where $\ell_1,\cdots, \ell_d$ are  linearly independent degree $1$ polynomials chosen uniformly at random and $a$ is randomly chosen from $\F_3$. Let $\rho_{r,d}(\alpha)$ be the eigenvalue of $S_{r,d}$ corresponding to $\chi_\alpha$, for $\alpha \in \Ef_r$.
\end{enumerate}
\end{definition}

 Now we will list some known bounds of the eigenvalues of the above operators. It is easy to see that $\lambda$ is a constant $<1$. Since $\F_3^R$ can be identified with $\Ef_r$, $T^{\otimes R}$ can be identified with $T_r$. Hence we have the following lemma.
\begin{lemma}
$$\left| \lambda_r(\alpha) \right| \leq |\lambda|^{|\alpha|}.$$
\end{lemma}

\begin{lemma}\label{lem:eigvalT}
For $\alpha \in \Lambda_{r,2d}$,
\begin{equation}
|\lambda_{r,d}(\alpha)| \begin{cases}
= |\lambda_r(\alpha)| &\text{ if } |\alpha| \leq 3^{d/2} \\
\leq 3^{-3^{C_1d}} &\text{ otherwise. } 
\end{cases}
\end{equation}
\end{lemma}
\begin{proof}
The first case follows from the fact that a random element $\eta$ according to $\mu_{r,2d}$ (the distribution that defines $T_{r,d}$) is $3^{d/2}$-wise independent (see \lref[Lemma]{lem:low-deg-local-ind}) as a string of length $3^r$ over alphabet $\F_3$. The latter case follows from \lref[Lemma]{lem:cor-rand-square}.
\end{proof}

We will derive bounds on the eigenvalues of $S_{r,d}$ using the results of Haramaty~\etal~\cite{HaramatySS2013}. Haramaty \etal\ analyses the following test for checking whether a polynomial is of degree $2r-2d - 1$: Choose a random affine subspace $S$ of dimension $r-d$ and check if the polynomial is of degree $2r-2d-1$ on $S$. Note that for any $\alpha \in \Pe_{r,2r-2d-1}$ and subspace $S$ of dimension $r-d$, $\sum_{x \in S} \alpha(x) = 0$. Hence this test is equivalent to choosing $\ell_1, \cdots \ell_d \in \Pe_{r,1}$ that are linearly independent and checking whether $\langle \alpha, \prod_{i=1}^d (\ell_i-1)(\ell_i-2) \rangle \neq 0$. Haramaty \etal\ proved the following lemma.
\begin{lemma}\label{lem:hss}
There exists constants $C_1,C_2$ such that
$$ \Pr_{\ell_i}\left[\langle \alpha, \prod_{i=1}^d (\ell_i-1)(\ell_i-2) \rangle = 0\right] \leq \max\left\{ 1- \frac{C_1 \Delta(\alpha, \Pe_{r,2r-2d-1})}{ 3^d} , C_2 \right\} $$
where $\ell_1,\cdots, \ell_i \in \Pe_{r,1}$ are random linearly independent polynomials.
\end{lemma}

\begin{lemma}\label{lem:eigvalS}
There exists constants $C'_1, C'_2$ such that, for $\alpha \in \Lambda_{r,2d}$,
\begin{equation}
1- \frac{2|\alpha|}{3^d} \leq |\rho_{r,d}(\alpha)| \leq \max\left\{ 1- \frac{C'_1 \Delta(\alpha, \Pe_{r,2r-2d-1})}{ 3^d} , C'_2 \right\} 
\end{equation}
\end{lemma}
\begin{proof}
First we will prove the lower bound. By definition 
$$\rho_{r,d}(\alpha) = \E_{\ell_i,a} \omega^{a \cdot \sum_{x} \alpha(x)  \prod_{i=1}^d (\ell_i(x)-1)(\ell_i(x)-2)}.$$
For any $x$ in support of $\alpha$, the probability that $\prod_{i=1}^d (\ell_i(x)-1)(\ell_i(x)-2)\neq 0$ is $1/3^d$. Hence by union bound, $\prod_{i=1}^d (\ell_i(x)-1)(\ell_i(x)-2)= 0$ for every $x$ in support of $\alpha$ with probability $1-|\alpha|/3^d$ and when this happens the expectation is $1$. Also note that the quantity inside the expectation has absolute value $1$.

For proving the upper bound we will use \lref[Lemma]{lem:hss}. Let $p_{\textsf{acc}}$ be the probability mentioned in \lref[Lemma]{lem:hss}. Then
$$ \rho_{r,d}(\alpha) = \E_{\ell_i,a} \omega^{a\langle \alpha, \prod_{i=1}^d (\ell_i-1)(\ell_i-2) \rangle} = p_{\textsf{acc}} + \frac{1-p_{\textsf{acc}}}{2} (\omega+\omega^2) = \frac{3}{2} p_{\textsf{acc}} - \frac{1}{2}.$$
From the above equation and  \lref[Lemma]{lem:hss}, the constants $C'_1,C'_2$ can be obtained.
\end{proof}

\begin{lemma}
\label{lem:low-deg-noise}
For $A,B: \Pe_{r,2d} \rightarrow [0,1]$, let $A' := S^t_{r,d}A$ and similarly define $B'$. Then 
$$\left|  \langle A, T_{r,d} B \rangle - \langle A', T_{r,d}  B'\rangle \right| \leq 2dt / 3^d $$
\end{lemma}

\begin{proof}

\begin{align*}
\left| \langle A, T_{r,d} B \rangle - \langle A', T_{r,d} B' \rangle \right| &\leq \left| \langle A, T_{r,d} B \rangle - \langle  A, T_{r,d} B' \rangle\right| \\
&+  \left| \langle A, T_{r,d}  B' \rangle -  \langle A', T_{r,d} B' \rangle \right| \\
 &= \left|  \langle  A-\E A, T_{r,d} (1-S^t_{r,d})(B-\E B) \rangle \right|\\
&+\left|  \langle   T_{r,d} (1-S^t_{r,d})(A-\E A), B'-\E B' \rangle \right|\\
&\leq  \|T_{r,d} (1-S^t_{r,d})(B-\E B)\| + \|T_{r,d} (1-S^t_{r,d})(A-\E A)\|\\
&\leq 2td/3^d
\end{align*}

The last step follows from the fact that the operators $T_{r,d} ,(1-S^t_{r,d})$ have the same set of eigenvectors and the largest eigenvalue in absolute value of $T_{r,d}(1-S^t_{r,d})$ is $2td/3^d$ from \lref[Lemma]{lem:eigvalT} and \lref[Lemma]{lem:eigvalS}.
\end{proof}

\noindent
\lref[Theorem]{thm:derand-ds} will follow from the following lemma.

\begin{lemma}
\label{lem:key-lem}
$\forall \epsilon > 0, \exists k = O(1/\epsilon^2) , d = O(\log (1/\epsilon))$ such that the following holds:
 if $A,B: \Pe_{r,2d} \rightarrow [0,1]$ then $\exists \mathcal A,\mathcal B :\Ef_r \rightarrow [0,1]$ such that 
\begin{enumerate}
\item $\left| \E A - \E \mathcal A\right| , \left| \E B - \E \mathcal B\right| \leq \epsilon $,
\item For all $x \in \F_3^r, k' \leq k$, 
$$\Inf_x^{\leq k'}(\mathcal A) \leq \Inf_x^{\leq k'}(A) + \epsilon$$
$$\Inf_x^{\leq k'}(\mathcal B) \leq \Inf_x^{\leq k'}(B) + \epsilon$$
\item $\left| \langle A, T_{r,d} B\rangle - \langle \mathcal A, T_r \mathcal B \rangle \right| \leq \epsilon$.

\end{enumerate}
\end{lemma}

\begin{proof}[Proof of {\lref[Theorem]{thm:derand-ds}}]
We will show that if \lref[Theorem]{thm:derand-ds} is false then \lref[Theorem]{thm:ds} is also false. First using \lref[Lemma]{lem:key-lem} with parameter $\epsilon = \mu^{O(1)}$, we obtain functions $\mathcal{ A, B}: \Ef_r \rightarrow [0,1]$ such that
\begin{enumerate}
\item $\E \mathcal A ,\E \mathcal B \geq \mu -\epsilon$ ,
\item For all $x \in \F_2^r, k' \leq k$, 
$$\Inf_{x}^{\leq k'}(\mathcal A) \leq \delta +\epsilon ~\text{ and }~\Inf_{x}^{\leq k'}(\mathcal B) \leq \delta  +\epsilon$$
\item $\left|  \langle \mathcal A, T_r \mathcal B \rangle \right|  \leq |\langle A, T_{r,d} B\rangle|+  \epsilon$.

\end{enumerate}

Now applying \lref[Theorem]{thm:ds} to the functions $\mathcal{ A,B}$, we obtain that  $\left|  \langle \mathcal A, T_r \mathcal B \rangle \right| \geq \delta'$, where $\delta' = \mu^{O(1)}$. Hence $\left|  \langle  A, T_{r,d}  B \rangle \right| \geq \delta' -\epsilon$, and we set the parameters $\delta = \delta' -\epsilon$, $d= O(\log 1/\mu)$ and $k = O(\log 1/\mu)$.

\end{proof}

\subsection{Proof of \texorpdfstring{\lref[Lemma]{lem:key-lem}}{Lemma 4.8}}
For proving \lref[Lemma]{lem:key-lem}, crucially use the following lemma by Kane and Meka~\cite{KaneM2013}.
\begin{lemma}\label{lem:km}Let $\xi:\R \rightarrow \R_{+}$ be the function $\xi(x) := \left(\max \{ -x, x-1,0 \}\right)^2$. For any parameters $k\in\mathbb{N}$ and $\epsilon \in (0,1)$, there is a $d = O(\log (k /\epsilon))$ such that the following holds: If the polynomial $P:\Ef_r \rightarrow \R$ satisfies $\|P\| \leq 1$ and $\widehat P(\alpha)=0$ for $\alpha \in \Lambda_{r,d}$ such that $|\alpha| > k$, then
$$\left| \E_{f \in \Ef_r } \xi(P(f)) - \E_{f\in \Pe_{r,d}} \xi(P(f)) \right| \leq \epsilon.$$
\end{lemma}

\begin{remark}
\label{rem:kane-meka}
For proving \lref[Lemma]{lem:km}, a generalization of ~\cite[Lemma 4.1]{KaneM2013} to polynomials of the form $P:\{ 1,\omega, \omega^2\}^R \rightarrow \R$ is required. However, we observe that the polynomials we consider are real-valued $P:\Ef_r \rightarrow \R$ and hence satisfy $\widehat P(\alpha) = \overline P(-\alpha)$. 

Using this observation, the proof of  ~\cite[Lemma 4.1]{KaneM2013} generalizes to our setting (the above property is preserved throughout the proof). The result of~\cite{KaneM2013} also requires an earlier result of Diakonikolas, Gopalana, Jaiswal, Servedio, and Viola~\cite{DiakonikolasGJSV2010} on fooling Linear Threshold functions (LTFs) with sample spaces of bounded independence. This proof also goes through for Thresholds of real-valued linear functions defined on variables that are uniformly distributed in $\{1,\omega,\omega^2\}$. \footnote{Such a function is the sign of a ``linear polynomial'' of the form $\left(\sum_{i=1}^R \alpha_i x_i + \overline{\alpha_i x_i}\right) - \theta$ for $\theta\in\mathbb{R}$.}
\end{remark}

\begin{proof}[Proof of {\lref[Lemma]{lem:key-lem}}]
Let $t = \frac{3^{d} \log (10/\epsilon)}{  2k}$, and $A_1 = S^t_{r,d} A, B_1 = S^t_{r,d} B$. Then from \lref[Lemma]{lem:low-deg-noise}
\begin{equation}
\left| \langle A, T_{r,d} B\rangle - \langle A_1, T_{r,d} B_1 \rangle \right| \leq 2dt/3^d
\end{equation}
and similarly for $B_1$. Let $k$ be a number $<3^{d/2}$ and $A_2 = \Re(A_1^{\leq k})$.  Using the fact that $A_1$ is real valued,
\begin{equation}
\| A_1 -  A_2 \| \leq \| A_1 - A_1^{\leq k} \| \leq (1-2k/3^d)^t \leq e^{-2tk/3^{d}} = \epsilon/10
\end{equation}

Let $A_3:\Ef_r \rightarrow \R$ be defined as $A_3 := \Re((A_1^{\leq k})')$ where $(A_1^{\leq k})'$ is the lift of $A_1^{\leq k}$. Since a random degree $d$ polynomial is $3^{d/2}$-wise independent,
\begin{equation}
\langle A_2, T_{r,d} B_2 \rangle = \langle A_3, T_r B_3 \rangle
\end{equation} 

 Note that $A_3$ may not be a $[0,1]$-valued function. But since $A$ is $[0,1]$-valued, so is $A_1$. Let $\xi:\R \rightarrow \R_{+}$ be the function $\xi(x) := \left(\max \{ -x, x-1,0 \}\right)^2$. Notice that $\E_f \xi \circ A(f) $ gives the $\ell^2_2$ distance of $A$ from $[0,1]$-valued functions. Using  \lref[Lemma]{lem:km}, for $d = O(\log (k/\epsilon))$,
\begin{equation}
\left| \E_{f \in \Pe_{r,2d}} \xi (A_2(f)) -  \E_{f \in \Ef_r} \xi (A_3(f)) \right| \leq \epsilon/10
\end{equation}
and similarly for $B_2$. Hence there exists functions $\mathcal {A,B}: \Ef_r \rightarrow [0,1]$ such that
\begin{enumerate}
\item $\left| \E A - \E \mathcal A\right| \leq ||A_1' - \mathcal A||  \leq \epsilon $ (similarly for $B$),
\item For all $x \in \F_3^r, k' \leq k$, $\Inf_x^{\leq k'}(\mathcal A) \leq \Inf_x^{\leq k'}(A) + \epsilon$
(similarly for $B$),
\item $\left| \langle A, T_{r,d} B\rangle - \langle \mathcal A, T_r \mathcal B \rangle \right| \leq \epsilon$.

\end{enumerate}

\end{proof}

{\small
\ifconf

\else
\bibliographystyle{../prahladhurl}
\fi
\bibliography{stacs20Dinur-bib}
}

\ifconf
\else
\appendix

\section{Hardness of Graph Coloring}

In this section we prove \lref[Theorem]{thm:col-hard}. Let $G =(U,V,R,E,\Pi)$ be a unique games cover instance with label set $R=\F_3^r$ and the constraints $\pi$ are full rank linear transformations. We will construct a graph $\mathcal{G=(V,E)}$ with $\mathcal V= V\times \Pe^r_{2d}$, where $d$ is a parameter to be fixed later. Let $T_{r,d}$ be the operator in \lref[Definition]{def:ops}. There is an edge in $\mathcal G$ between $(v,f)$ and $(w,g)$  if there is a $u\in U$ such that $(u,v), (u,w) \in E$ and $T_{r,d}(f\circ\pi^{-1}_{u,v},g \circ \pi^{-1}_{u,w}) >0$, where $\pi_{u,v}$ is the full rank linear map that maps a label of $v$ to label of $u$.

\begin{lemma}[Completeness]\label{lem:completeness}
If $G$ belongs to the YES case of $(c,s,r)$-UG Conjecture then $\cal G$ has a induced subgraph of relative size $1-c$ that is $3$-colorable.
\end{lemma}
\begin{proof}
Suppose the label cover instance has a labeling $\ell: V \rightarrow \F_3^r$ and a set $S \subseteq V, |S| = (1-c)|V|$, such that $\ell$ satisfies all the edges incident on vertices in $S$. We will show that $A_v(f):=f(\ell(v))$ for $v\in V$, is a $3$-coloring for the induced subgraph of $\cal G$ on the set $S \times \Pe_{r,2d}$. For any $u \in U, v,w \in S$ having edges $(u,v),(u,w)\in E$, consider the edge $((v,f),(w,g)) \in \mathcal E$. The colors given to the end points are $f(\ell(v))$ and $g(\ell(w))$. Since $T_{r,d}(f\circ\pi^{-1}_{u,v},g\circ \pi^{-1}_{u,w}) >0$, 
$$g\circ \pi^{-1}_{u,w} = f\circ \pi^{-1}_{u,v} + a(p^2 +1) \text{ for some  } p \in \Pe^r_d, a \in \{1,2\}.$$
 So $f(\ell(v))= f \circ \pi_{u,v}^{-1}(\ell(u))  \neq g\circ \pi_{u,w}^{-1}(\ell(u)) = g (\ell(w))$.

\end{proof}

\begin{lemma}[Soundness]\label{lem:soundness}
If $G$ belongs to the NO case of $(c,s,r)$-UG Conjecture, $\cal G$ has an independent set of relative size $\mu$ and $d= O(\log 1/\mu)$ then
$ \mu  \leq \poly(s(n)).$
\end{lemma}
\begin{proof}
Let $I_v:\Pe_{r,2d} \rightarrow \{0,1\}$ be the indicator function of $I$ restricted to vertices in $\mathcal V$ corresponding to $v\in V$. Let $J=\{v \in V:  \E_{f \in \Pe_{r,2d}} I_v(f) \geq \mu/2\}$. Then we have that $|J|/|V| \geq \mu/2$. For $v\in J$, let $L(v) = \{x \in \F_3^{r}: \Inf_x^{\leq k}(I_v) > \delta \}$ where $\delta = \poly(\mu), k = O(\log 1/\mu)$ are parameters from \lref[Theorem]{thm:derand-ds}. Note that $|L(v)| \leq k/\delta$, since the sum of all degree $k$ influences is at most $k$. 
\begin{claim}\label{claim:sound}
Let $v,w \in J$ and $(u,v),(u,w) \in E$. Then there exists $a \in L(v), b \in L(w)$ such that $\pi_{u,v}(a)= \pi_{u,w}(b)$.
\end{claim}
\begin{proof}
Let $A(f):= I_v(f\circ\pi^{-1}_{u,v})$, $B(g):= I_w(g\circ \pi^{-1}_{u,w})$. Since $I$ is an independent set, if $(v,f\circ\pi^{-1}_{u,v}), (w,g\circ\pi^{-1}_{u,w}) \in I$, then  $T_{r,d}(f\circ\pi^{-1}_{u,v},g\circ\pi^{-1}_{u,w}) =0$, which gives that
\begin{equation}
\langle A, T_{r,d} B\rangle = 0
\end{equation}
From \lref[Theorem]{thm:derand-ds}, there is some $c \in \F_3^r$ such that $\Inf^{\leq k}_c( A),\Inf^{\leq k}_c(B) > \delta$, which gives that $\pi_{u,v}^{-1}(c) \in L(v)$ and $\pi_{u,w}^{-1}(c) \in L(w)$. 
\end{proof}
Now consider the randomized partial labeling $L'$ to $G$, where for $v \in J$, $L'(v)$ is chosen randomly from $L(v)$ and for $u\in U$,  choose a random neighbor $w \in J$ (if it exists), a random label $a \in L(w)$ and set $L(u)= \pi_{u,w}^{-1}(a)$. For any $v \in J$,  any edge $(u,v)$, the probability of it being satisfied by $L'$ is $\mu^2/k^2 = \poly(\mu)$, because of \lref[Claim]{claim:sound}.

\end{proof}

\begin{proof}[Proof of Theorem \ref{thm:col-hard}]
The size of $\mathcal G$ denoted by $N$ is at most $  n 3^{r^{O(d)}}$. Substituting $r= 2^{O(\sqrt{\log \log n})}, d = \log 1/\mu \leq O(\sqrt{\log \log n})$, we get that $N = \poly(n)$ and hence the reduction is polynomial time.
\end{proof}

\ifarxiv
\else
\section{Towards an even better ``majority is stablest'' result}\label{app:conj}

The following is a more ambitious variant of Conjecture 5.13 in the work of Barak et al.~\cite{BarakGHMRS2012}.

\begin{conjecture}
\label{conj:PTF}
For any parameters $k\in\mathbb{N}$ and $\eta\in (0,1)$, there is an $\ell = (k\lg \frac{1}{\eta})^{O(1)}$ such that the following holds for any $\ell$-wise independent distribution $\mathcal{D}$ over $\Ef_r$. Let $A:\Ef_r\rightarrow \mathbb{R}$ be a polynomial of degree at most $k$ (i.e. $\widehat{A}(\alpha) = 0$ for all $\alpha$ such that $|\alpha| > k$) and $\theta\in\mathbb{R}$ be such that $\prob{f\in\Ef_r}{A(f)\geq \theta}\geq \eta$. Then, we have $\prob{f\sim \mathcal{D}}{A(f) \geq \theta}\geq \frac{\eta^{O(1)}}{2^{O(k)}}$.
\end{conjecture}

Assuming the above conjecture, we can prove the following claim, which allows us to prove an improved version of \lref[Lemma]{lem:km} for our application. The rest of the proof also needs some simple modifications in order to improve the parameters in Theorem~\ref{thm:derand-ds}, but these details are omitted. 

\begin{claim}
\label{clm:using-conj}
Assume Conjecture~\ref{conj:PTF}. Fix any $k\in\mathbb{N}$ and $\nu\in (0,1)$. There is a $d = O(\lg k + \lg\lg(\frac{1}{\nu}))$ such that the following holds. Let $A:\Ef_r\rightarrow \mathbb{R}$ be a degree $\leq k$ polynomial such that $||A||\leq 1$. Then, for $\xi$ as defined in the \lref[Lemma]{lem:km}, we have 
\[
\avg{f\in \Ef_r}{\xi(A(f))} \geq \nu \Rightarrow \avg{f\in \Pe_{r,d}}{\xi(A(f))} \geq \frac{\nu^{O(1)}}{2^{O(k)}}
\]
\end{claim}

\begin{proof}
Assume that we have $\avg{f\in \Ef_r}{\xi(A(f))} \geq \nu$. Recall that the uniform distribution over $\Pe_{r,d}$ is $3^{d/2}$-wise independent. We will choose $d$ such that $3^{d/2}\geq \max\{\ell,2k\}$ for some $\ell = (k\lg\frac{1}{\nu})^{O(1)}$ which will be made more explicit later on in the proof. Clearly, any such choice of $d$ satisfies $d = O(\lg k + \lg \lg (\frac{1}{\nu}))$  as claimed. For such a choice of $d$, we will show that $\avg{f\in \Pe_{r,d}}{\xi(A(f))}$ is as large as claimed.

Firstly, note that 
\[
\xi(A(f)) = A(f)^2\mathbf{1}_{A(f)\leq 0} + (A(f)-1)^2 \mathbf{1}_{A(f)\geq 1}
\]
Hence, if $\avg{f\in \Ef_r}{\xi(A(f))} \geq \nu$, we must have either:
\[
\avg{f\in \Ef_r}{A(f)^2\mathbf{1}_{A(f)\leq 0}}\geq \frac{\nu}{2} \qquad \text{OR}\qquad 
\avg{f\in \Ef_r}{(A(f)-1)^2 \mathbf{1}_{A(f)\geq 1}}\geq \frac{\nu}{2}
\]

For the rest of the proof, we assume that the latter holds (the other case is similar). Note furthermore, that for any choice of $\theta \geq 1$, we have
\begin{align*}
\avg{f\in \Ef_r}{(A(f)-1)^2 \mathbf{1}_{A(f)\geq 1}} &= \avg{f\in \Ef_r}{(A(f)-1)^2 \mathbf{1}_{A(f)\in [1,\theta]}} + \avg{f\in \Ef_r}{(A(f)-1)^2 \mathbf{1}_{A(f)\geq \theta}}\\
&\leq (\theta-1)^2 + \avg{f\in \Ef_r}{(A(f)-1)^2 \mathbf{1}_{A(f)\geq \theta}}
\end{align*}

Choosing $\theta= (1+\frac{\sqrt{\nu}}{2})$ and plugging it into the inequality above, we see that
\[
\avg{f\in \Ef_r}{(A(f)-1)^2 \mathbf{1}_{A(f)\geq 1}} \leq \frac{\nu}{4} + \avg{f\in \Ef_r}{(A(f)-1)^2 \mathbf{1}_{A(f)\geq \theta}}
\]
Since the left hand side of the above inequality is at least $\frac{\nu}{2}$, we see that
\begin{equation}
\label{eq:claim-conj-1}
\avg{f\in \Ef_r}{(A(f)-1)^2 \mathbf{1}_{A(f)\geq \theta}} \geq \frac{\nu}{4}
\end{equation}
We can further upper bound the left hand side of the above using the Cauchy-Schwarz inequality as follows. 
\begin{align*}
\avg{f\in \Ef_r}{(A(f)-1)^2 \mathbf{1}_{A(f)\geq \theta}} &\leq \sqrt{\avg{f\in \Ef_r}{(A(f)-1)^4}}\cdot \sqrt{\avg{f\in \Ef_r}{\mathbf{1}_{A(f)\geq \theta}}}\\
&= \sqrt{\avg{f\in \Ef_r}{(A(f)-1)^4}}\cdot\sqrt{\prob{f\in \Ef_r}{A(f)\geq \theta}}
\end{align*}

We can now upper bound the first term on the right hand side of the above inequality by using Lemma~\ref{thm:hyp} to obtain
\[
\avg{f\in \Ef_r}{(A(f)-1)^2 \mathbf{1}_{A(f)\geq \theta}}\leq 2^{O(k)}\avg{f\in \Ef_r}{(A(f)-1)^2}\cdot\sqrt{\prob{f\in \Ef_r}{A(f)\geq \theta}}
\]

By the above inequality and (\ref{eq:claim-conj-1}) we have
\begin{equation}
\label{eq:claim-conj-2}
\avg{f\in \Ef_{r}}{(A(f)-1)^2}\cdot\sqrt{\prob{f\in \Ef_r}{A(f)\geq \theta}} \geq \frac{\nu}{2^{O(k)}}.
\end{equation}

The first term above may be further upper bounded by a constant since we have
\[
\avg{f\in \Ef_{r}}{(A(f)-1)^2} \leq \avg{f\in \Ef_{r}}{(A(f))^2} + 2\avg{f\in \Ef_{r}}{|A(f)|} + 1 \leq \avg{f\in \Ef_{r}}{(A(f))^2} + 2\sqrt{\avg{f\in \Ef_{r}}{(A(f))^2}} + 1 \leq 4.
\]
where the second inequality follows from the Cauchy Schwarz inequality and the last inequality since by the hypothesis of the claim, we have $||A|| = \sqrt{\avg{f\in \Ef_{r}}{(A(f))^2}}\leq 1$. Thus, by (\ref{eq:claim-conj-2}), we get 
\[
\prob{f\in \Ef_r}{A(f)\geq \theta} \geq \frac{\nu^2}{2^{O(k)}}.
\]

Now, from the statement of Conjecture~\ref{conj:PTF} applied with parameters $k$ as in the statement of the claim and $\eta$ being the right hand side of the above inequality, we know that there is an $\ell=(k\lg(\frac{1}{\eta}))^{O(1)} = (k\lg (\frac{1}{\nu}))^{O(1)}$ such that for any distribution $\mathcal{D}$ that is $\ell$-wise independent, we have 
\[
\prob{f\sim \mathcal{D}}{A(f)\geq \theta} \geq \frac{\eta^{O(1)}}{2^{O(k)}} = \frac{\nu^{O(1)}}{2^{O(k)}}.
\]

In particular, the above holds for $\mathcal{D}$ being the uniform distribution over $\Pe_{r,d}$ --- which is $3^{d/2}$-wise independent --- as long as we ensure $3^{d/2} \geq \ell$ (note that this agrees with our choice of $d$ earlier on). Hence, we have the statement of the claim.

\end{proof}

\fi
\fi

\end{document}